\documentclass[11pt,reqno]{amsart}
\usepackage[T1]{fontenc}
\usepackage{lmodern,amsmath,amssymb,amsthm,mathtools,mathrsfs,microtype}
\usepackage[paper=letterpaper,left=1in,right=1in,top=.9in,bottom=.9in]{geometry}
\usepackage{graphicx,tikz}
\usetikzlibrary{arrows.meta,calc}
\usepackage[hidelinks]{hyperref}
\microtypesetup{expansion=false}
\allowdisplaybreaks[2]
\newtheorem{theorem}{Theorem}[section]
\newtheorem{proposition}[theorem]{Proposition}
\newtheorem{lemma}[theorem]{Lemma}
\newtheorem{corollary}[theorem]{Corollary}
\theoremstyle{remark}

\newcommand{\T}{\mathbb T}\newcommand{\Z}{\mathbb Z}
\newcommand{\R}{\mathbb R}\newcommand{\C}{\mathbb C}
\newcommand{\meas}{\operatorname{meas}}

\newcommand{\Id}{I}\newcommand{\dd}{\,d}
\newcommand{\si}{\sigma}\DeclareMathOperator{\Tr}{Tr}
\DeclareMathOperator{\rank}{rank}\DeclareMathOperator{\Var}{Var}
\DeclareMathOperator{\dist}{dist}
\DeclareMathOperator{\supp}{supp}\DeclareMathOperator{\diag}{diag}

\newcommand{\norm}[1]{\lVert#1\rVert}
\newcommand{\lbox}{\underline{\dim}_{\mathrm B}}
\newcommand{\ubox}{\overline{\dim}_{\mathrm B}}
\newcommand{\dimH}{\mathrm{dim_H}}\newcommand{\pdim}{\dim_{\mathrm P}}

\newcommand{\be}{\begin{equation}}
\newcommand{\ee}{\end{equation}}

\newcommand{\wt}{\widetilde}
\newcommand{\wh}{\widehat}
\newcommand{\de}{\delta}

\newcommand{\al}{\alpha}
\newcommand{\bt}{\beta}
\newcommand{\ga}{\gamma}

\newcommand{\Om}{\Omega}
\newcommand{\lb}{\lambda}

\renewcommand{\th}{\theta}

\newcommand{\ep}{\varepsilon }

\newcommand{\spec}{{\mathrm{spec}}}

\newcommand{\mg}{}

\title[Fractal dimensions for EHM]{
Hidden singularity and fractal dimensions\\for extended Harper operators}
\author{S. Becker}
\address{Bocconi University, Milan, Italy}
\author{S. Jitomirskaya}
\address{University of California, Berkeley, USA, and University of California, Irvine, USA}
\author{I. Krasovsky}
\address{Imperial College London, UK, and SISSA, Trieste, Italy}
\date{}
\hypersetup{pdftitle={Hidden singularity and fractal dimensions for  extended Harper operators},
 pdfauthor={S. Becker, S. Jitomirskaya, I. Krasovsky}}
\RequirePackage{xcolor}
\begin{document}
\begin{abstract}
We prove that, for every irrational frequency, the spectra of the extended Harper family throughout
the zero-measure regime have finite $1/2$-dimensional Hausdorff measure and lower box dimension at
most $1/2$. Both {\mg numbers} improve to $1/3$ on the isotropic critical boundary. The proof is based on a family of unitary transformations, that specializes to the chiral transform of \cite{bjk} for the case of the critical almost Mathieu. This family allows to make every operator in
this regime spectrally equivalent to a {\mg singular Jacobi operator.} Using these singular representatives, we {\mg apply} the method of \cite{bjk} to construct covers by $O(q_n)$ intervals of total
length $O(q_n^{-1})$, improving to $O(q_n^{-2})$ on the isotropic boundary, {\mg where $q_n$ is the denominator of a canonical approximant to $\alpha$.} The covers give
upper box bounds and, with packing estimates, examples where the box-counting dimension
fails to exist.
\end{abstract}
\maketitle

\section{Introduction}\label{sec:models}
The extended Harper model describes an electron on a square lattice in a
uniform perpendicular magnetic field, with hopping between nearest 
and {\mg next-nearest} 
neighbors. It is thus a natural next model
after the Harper model in the study of magnetic lattice spectra {\cite{HTHK}}. 
The latter goes back to the work of Peierls and Harper
\cite{Peierls,Harper}, and its dependence on the magnetic flux is beautifully
depicted by the Hofstadter butterfly \cite{Hofstadter}. Choosing the Landau
gauge and taking a partial Fourier transform in one lattice direction gives
a family of one-dimensional quasiperiodic Jacobi operators, also referred to
as the extended Harper model (EHM). Writing $2\pi\alpha$ for the flux through
a square and $\theta\in\T=\mathbb{R}/\mathbb{Z}$ for the transverse
quasimomentum, for $\lambda=(\lambda_1,\lambda_2,\lambda_3)\in[0,\infty)^3$
the extended Harper operator on $\ell^2(\Z)$ is
\begin{align}
 (H_{ \alpha,\theta ,\lambda }u)_m
 &=\overline{c_{\lambda,\alpha}(\theta+(m-1)\alpha)}u_{m-1}
   +c_{\lambda,\alpha}(\theta+m\alpha)u_{m+1}
   \notag\\[-2pt]
 &\hspace{2em}+2\cos(2\pi(\theta+m\alpha))u_m,  \notag \\
 c_{\lambda,\alpha}(t)
 &=\lambda_1e^{-2\pi i(t+\alpha/2)}
   +\lambda_2+\lambda_3e^{2\pi i(t+\alpha/2)}.\label{eq:c}
\end{align}
For irrational $\alpha$, its spectrum is independent of $\theta$; we denote
it by   $\Sigma_{\alpha,\lambda}$ . For $\lambda=(0,\lambda_2,0)$, with
$\lambda_2>0$, EHM becomes the rescaled almost Mathieu operator, {\mg which arises by}
allowing only nearest-neighbor hopping in the two-dimensional magnetic
Laplacian and {\mg which features} a basic phase transition at $\lambda_2=1$.
Allowing the additional hoppings leads to a considerably richer phase
diagram {\mg (cf. Figure \ref{fig:parameter-regions})}, with several critical boundaries between regions I,II, and III, and the entire region III being self-dual; see
\cite{MZ,HTHK,AJM}. The triangular lattice is also a special case of extended
Harper \cite{InoKohmoto}.

In mathematics, spectral theory of the extended Harper model has been a
popular domain for the study of validity and limitations of arguments
originally developed for the almost Mathieu family, e.g.\
\cite{JKS,JM12,JM13,HanDry}. The main interest in this regard comes from two aspects
that differ significantly from almost Mathieu: the nontrivial self-dual
region III, where the usual duality between supercritical and subcritical
regimes is not available, and the set of parameters where the corresponding
Jacobi matrix becomes singular, that is $c_{\lambda,\alpha}(\theta)$ has a zero on {\mg a period}, limiting applicability of methods based on
nonsingularity. 

The phase diagram of spectral types is by now well
understood \cite{AJM}. Various questions, in particular, about
the geometry of the spectrum, however, remain.

In particular,  fractal dimensions of the spectrum have been largely unstudied. The question, of course, only makes sense when the spectrum has measure zero, a regime,  that has remained rather mysterious, being not amenable to {\mg either reducibility or} localization methods. For any irrational $\alpha$, a direct corollary of
{ Proposition 3.1}, Theorems~3.4, 3.5, and~4.1 of {\cite{AJM}} and Aubry duality is that,
{\mg writing $\Lambda=\lambda_1+\lambda_3$,
\[ \label{lines}
|\Sigma_{\alpha,\lambda}|=0
\quad\Longleftrightarrow\quad
\lambda\in\mathcal C_0=
\begin{aligned}[t]
&\{\lambda:\Lambda=1\geq\lambda_2\}\\
&{}\cup\{\lambda:\lambda_2=1\geq \Lambda\}\\
&{}\cup\{\lambda:\Lambda=\lambda_2\geq1\}\\
&{}\cup\{\lambda:\lambda_1=\lambda_3,\ \Lambda\ge\max\{1,\lambda_2\}\}.
\end{aligned}
\]}
The first three regions, denoted by $\mathcal L_{\mathrm I}$, $\mathcal L_{\mathrm{II}}$, and $\mathcal L_{\mathrm{III}}$, respectively, are precisely the critical boundaries
 where at least two of the three quantities
$1$, $\lambda_2$, and $\Lambda=\lambda_1+\lambda_3$ attain the $\max(1,\lambda_2,\Lambda)$.
The fourth, denoted by $\mathcal{I}_{\mathrm{iso}}$ is the isotropic\footnote{The isotropy means $\lambda_1=\lambda_3$.} 
part of the interior of region~III.

{We} distinguish the isotropic critical boundary
$$
\mathcal{B}_{\mathrm{iso}}
 =\{\lambda:\lambda_1=\lambda_3,\
               \lambda_1+\lambda_3=\max\{1,\lambda_2\}\},
$$
{\mg which is the union of intersections $\mathcal I_{\rm iso}\cap\mathcal L_j$, $j=I,II,III$. }

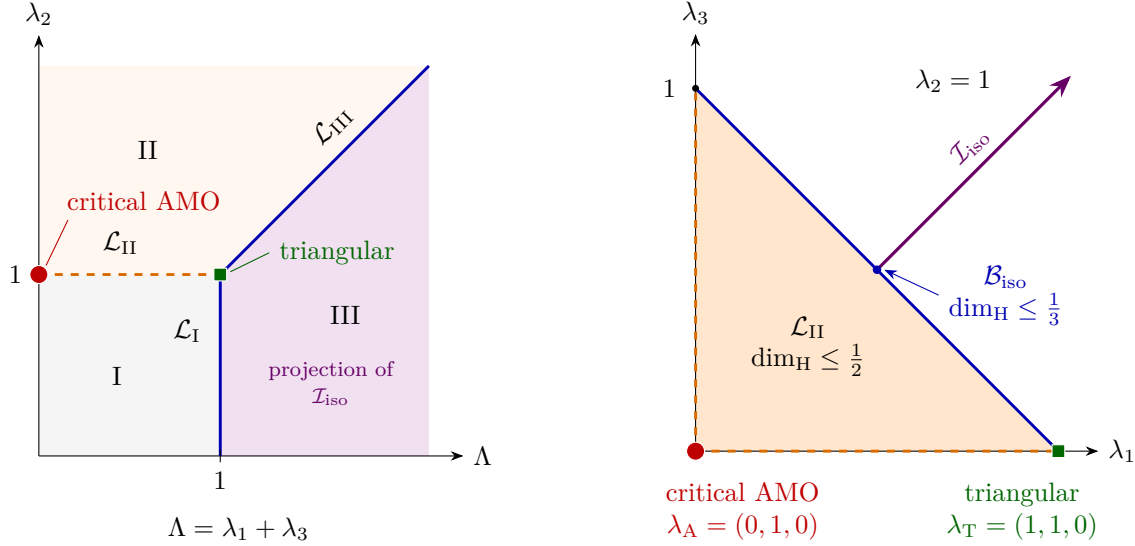
\begin{figure}[htbp]
 \centering
 \tikzset{
  amo/.style={circle,fill=red!75!black,draw=white,line width=.4pt,inner sep=2.4pt},
  triangular/.style={rectangle,fill=green!40!black,draw=white,line width=.4pt,inner sep=2.4pt}}
 \begin{minipage}[t]{.48\textwidth}
  \centering
  \begin{tikzpicture}[x=2.40cm,y=2.40cm,>=Stealth,font=\small]
   \fill[gray!8] (0,0) rectangle (1,1);
   \fill[orange!6] (0,1)--(1,1)--(2.15,2.15)--(0,2.15)--cycle;
   \fill[violet!12] (1,0)--(2.15,0)--(2.15,2.15)--(1,1)--cycle;
   \draw[->] (0,0)--(2.34,0) node[right] {$\Lambda$};
   \draw[->] (0,0)--(0,2.32) node[above] {$\lambda_2$};
   \draw[blue!70!black,line width=1.1pt] (1,0)--(1,1)--(2.15,2.15);
   \draw[orange!85!black,dashed,line width=1.1pt] (0,1)--(1,1);
   \draw (1,.035)--(1,-.035) node[below] {$1$};
   \draw (.035,1)--(-.035,1) node[left] {$1$};
   \node at (.43,.44) {$\mathrm I$};
   \node at (.60,1.69) {$\mathrm {II}$};
   \node at (1.69,.78) {$\mathrm {III}$};
   \node[align=center,font=\scriptsize,text=violet!80!black] at (1.61,.40)
       {projection of\\$\mathcal I_{\mathrm{iso}}$};
   \node[anchor=east] at (.95,.69) {$\mathcal L_{\mathrm I}$};
   \node[anchor=south] at (.45,1.055) {$\mathcal L_{\mathrm {II}}$};
   \node[rotate=45,anchor=south] at (1.70,1.74) {$\mathcal L_{\mathrm {III}}$};
   \node[amo] at (0,1) {};
   \draw[red!75!black,thin] (.04,1.05)--(.16,1.30);
   \node[anchor=south west,text=red!75!black] at (.10,1.30) {critical AMO};
   \node[triangular] at (1,1) {};
   \draw[green!40!black,thin] (1.055,1.015)--(1.27,1.12);
   \node[anchor=west,text=green!40!black] at (1.27,1.12) {triangular};
   \node[anchor=north] at (1.1,-.28) {$\Lambda=\lambda_1+\lambda_3$};
  \end{tikzpicture}
 \end{minipage}\hfill
 \begin{minipage}[t]{.48\textwidth}
  \centering
  \begin{tikzpicture}[x=4.80cm,y=4.80cm,>=Stealth,font=\small]
   \fill[orange!20] (0,0)--(1,0)--(0,1)--cycle;
   \draw[->] (0,0)--(1.11,0) node[right] {$\lambda_1$};
   \draw[->] (0,0)--(0,1.15) node[above] {$\lambda_3$};
   \draw[orange!85!black,dashed,line width=1.1pt] (0,0)--(1,0) (0,0)--(0,1);
   \draw[blue!70!black,line width=1.1pt] (1,0)--(0,1);
   \node[align=center] at (.31,.29)
       {$\mathcal L_{\mathrm {II}}$\\$\dimH\le\tfrac12$};
   \node at (.71,1.02) {$\lambda_2=1$};
   \draw[violet!80!black,line width=1.2pt,->] (.5,.5)--(1.035,1.035)
       node[pos=.56,above,sloped] {$\mathcal I_{\mathrm{iso}}$};
   \node[amo] at (0,0) {};
   \node[anchor=north,align=center,text=red!75!black] at (.13,-.06)
       {critical AMO\\$\lambda_{\mathrm A}=(0,1,0)$};
   \node[triangular] at (1,0) {};
   \node[anchor=north,align=center,text=green!40!black] at (.90,-.06)
       {triangular\\$\lambda_{\mathrm T}=(1,1,0)$};
   \fill (0,1) circle (1.3pt);
   \node[anchor=east] at (-.025,1) {$1$};
   \fill[blue!70!black] (.5,.5) circle (1.6pt);
   \draw[blue!70!black,<-] (.52,.50)--(.67,.43)
       node[right,align=center]
       {$\mathcal B_{\mathrm{iso}}$\\$\dimH\le\tfrac13$};
  \end{tikzpicture}
 \end{minipage}
 \caption{Left: the EHM regions projected onto $(\Lambda,\lambda_2)$, where
 $\Lambda=\lambda_1+\lambda_3$. The blue boundary is $\Lambda=\max\{1,\lambda_2\}$.
 The purple shading is the projection of $\mathcal I_{\mathrm{iso}}$;
 in the interior of region~III, only couplings with $\lambda_1=\lambda_3$ belong to $\mathcal C_0$.
 {\mg Right: projection onto $\lambda_2=1$.} The purple ray $\lambda_1=\lambda_3\ge1/2$ is the slice of
 $\mathcal I_{\mathrm{iso}}$. The entire shaded triangle
 $\mathcal L_{\mathrm {II}}\subset\mathcal C_0$ has {\mg Hausdorff dimension of the spectrum}
 at most $1/2$ for every irrational frequency. At the blue midpoint
 $(1/2,1,1/2)\in\mathcal B_{\mathrm{iso}}$, the bound improves to $1/3$.
 Red circles mark the critical AMO; green squares mark the triangular operator.}
 \label{fig:parameter-regions}
\end{figure}

For the almost Mathieu family, $\lambda=(0,\lambda_2,0),$ there is only one parameter with zero measure spectrum: $\lambda_2=1.$ Thouless conjectured that Hausdorff dimension is bounded by $1/2$ in this case, for all irrational $\alpha.$ 

In \cite{bjk} we developed a method to find tight spectral covers for {\it singular} Jacobi matrices, that lead to dimension bounds. Historically, singularity has been an obstacle to overcome when extending
various almost Mathieu arguments \cite{JM12,JM13,AJM,HYZ}. In contrast, in
\cite{bjk} we exploit singularity: while critical almost Mathieu is a Schrödinger operator, so certainly not a singular Jacobi
matrix,  we found what we call a chiral transform that gives it an isospectral singular Jacobi
representation, leading to a proof of the Thouless conjecture for every irrational frequency. 

 Here we show that the method of \cite{bjk} can be extended to the entire  zero-measure regime of the EHM, giving the same bound {\mg $1/2$ without arithmetic restrictions on $\al$.} Moreover,
on the isotropic critical boundary the bound improves to $1/3$.

Operators in {${\mathcal L}_{\mathrm {III}}$} and $\mathcal{I}_{\mathrm{iso}}$ are already singular, and operators in {${\mathcal L}_{\mathrm{I}}$} are made singular by the standard Aubry duality. For those, the problem reduces to finding good spectral covers. 

The region $\mathcal L_{\mathrm{II}}$, however, is more challenging. It is
the  self-dual triangle {\mg (see Figure \ref{fig:parameter-regions})}, that contains 
critical almost Mathieu and the isotropic triangular model as its boundary points, at {\mg $(0,1,0)$ and $(1, 1, 0)$, respectively}.  The additional challenge is to unlock its hidden singularity. We do it by finding a family of transforms  that singularize each $H_{\alpha,\theta,\lambda}$ with $\lambda\in \mathcal L_{\mathrm{II}}.$ The chiral transform of \cite{bjk} is one element in this family.

{\mg Intuitively}, the physical question of what survives when additional hopping is allowed
is related to the recently established analytic robustness of the almost Mathieu results \cite{GJY, gj, lxz} , although we note that criticality itself is not robust \cite{avilaglobal}.
Our results are global and complete in the parameter class considered here, with no
smallness assumptions or other restrictions.


\subsection{Hausdorff dimension}


We write   $\meas_t$  for $t$-dimensional Hausdorff measure,
$\dimH $ for Hausdorff dimension, and
$\underline{\dim}_{\mathrm{B}}$ and $\overline{\dim}_{\mathrm{B}}$ for lower
and upper box dimensions.

\begin{theorem}\label{thm:HD}
For every irrational $\alpha$ and every   $\lambda\in\mathcal{C}_0$ ,
$$
  \meas_{1/2}(\Sigma _{ \alpha ,\lambda })<\infty,
\qquad
\dimH \Sigma_{ \alpha ,\lambda }
\leq\underline{\dim}_{\mathrm{B}}\Sigma_{ \alpha ,\lambda }\leq\frac12.
$$
For $\lambda\in\mathcal{B}_{\mathrm{iso}}$, these conclusions improve to
$$
  \meas_{1/3}(\Sigma _{ \alpha ,\lambda })<\infty,
\qquad
\dimH \Sigma_{ \alpha ,\lambda }
\leq\underline{\dim}_{\mathrm{B}}\Sigma_{ \alpha ,\lambda }\leq\frac13.
$$
\end{theorem}

The almost Mathieu results of Theorem {\ref{thm:HD}} provide a strengthening of the Hausdorff  dimension  bound of \cite{bjk}{,
following from the covers proved there in Theorem~7.11} .
The $1/3$ bound concerns the bicritical boundary studied in
\cite{HTHK}, including the point $(1/2,1,1/2)$. {\mg A different} scaling at
this boundary was already apparent in the rational bandwidth analysis of
\cite{HTHK}; here it leads to a dimension bound for every irrational
frequency. {\mg In our argument, it is the higher-order zero of the 
hopping coefficient} that produces the improvement.

\subsection{Box-counting dimension}

Of course, a bound on Hausdorff dimension does not in general control
box-counting dimension. Already for critical almost Mathieu, the Hausdorff
dimension can be zero \cite{LastShamis,ALSZ} while the packing dimension is
one \cite{JZ}. Our argument also gives bounds on upper box dimension in
terms of the arithmetic properties of the frequency.

Let $p_n/q_n$ be the {\mg canonical approximants of $\alpha$, and set
{
\begin{equation}\label{beta}
\tau(\alpha)=\limsup_n\frac{\ln q_{n+1}}{\ln q_n}\in [1,\infty],
\qquad
\beta(\alpha)=\limsup_n\frac{\ln q_{n+1}}{q_n}\in [0,\infty].
\end{equation}
\par}}
  For almost every   $\alpha$,  $\tau(\alpha)=1$ and $\beta(\alpha)=0.$ We call $\alpha$ weakly Liouville when $\beta(\alpha)>0$ {\mg (and therefore $\tau(\al)=\infty$)}.

\begin{theorem}\label{thm:upper-box}
For every irrational $\alpha$ and every   $\lambda\in\mathcal{C}_0$ ,
$$
\dim_{\mathrm{P}}\Sigma_{ \alpha ,\lambda }
\leq\overline{\dim}_{\mathrm{B}}\Sigma_{ \alpha ,\lambda }
\leq\frac{\tau(\alpha)}{1+\tau(\alpha)}.
$$
For $\lambda\in\mathcal{B}_{\mathrm{iso}}$, the last bound improves to
$$
\overline{\dim}_{\mathrm{B}}\Sigma_{ \alpha ,\lambda }
\leq\frac{\tau(\alpha)}{2+\tau(\alpha)}.
$$
\end{theorem}

Here $\dim_{\mathrm{P}}$ denotes packing dimension.
When $\tau(\alpha)=1$, which holds for Lebesgue almost every $\alpha$,
these are upper box-dimension bounds of $1/2$ and $1/3$, respectively.

For weakly Liouville frequencies the situation is quite different.
Combining our   bounds on lower box dimension  with the packing estimates of
\cite{JZ,HYZ} gives

\begin{corollary}\label{cor:box-nonexistence}
{\mg Suppose $\beta(\alpha)>0$, and 
$$
\lambda\in  \{(0,1 , 0),(1,1,0)\}
\quad\text{or}\quad
\lambda\in\mathcal{I}_{\mathrm{iso}} { \cup \mathcal B_{\mathrm{iso}}}, \quad\lambda_2>0.
$$
Then
\be\label{noboxdim}
\underline{\dim}_{\mathrm{B}}\Sigma_{\alpha,\lambda}\leq\frac12
<1=\dim_{\mathrm{P}}\Sigma_{\alpha,\lambda}=\overline{\dim}_{\mathrm{B}}\Sigma_{\alpha,\lambda}.
\ee}
For $\lambda\in\mathcal{B}_{\mathrm{iso}}$ with $\lambda_2>0$, the lower box
bound improves to $1/3$. Thus box-counting dimension does not exist in all
these cases, including the critical almost Mathieu and triangular models.
\end{corollary}

This distinction is important when interpreting the dimension suggested by
rational approximants or numerical calculations. The small dimension seen
along one sequence of scales need not be the box-counting dimension of the
spectrum: the limit defining the latter may simply fail to exist.

{\mg In Section \ref{sec:covers} we find covers by intervals for the spectrum $\Sigma_{\al,\lb}$
when $\al$ is irrational and $\lb\in{\mathcal C}_0$.
Following the approach of \cite{bjk},
we first reduce EHM for $\lb\in{\mathcal C}_0$
to singular Jacobi form. Then we need to consider blocks obtained from the Jacobi matrix by replacing by zero hopping coefficients where they are small. 
The blocks $B(x)$ so obtained will have two dimensions, $q_n$ and $q_{n-1}$ and will depend on a phase parameter $x$ which ranges over a small interval of length of order $1/q_n$ (with removed hopping coefficients of order $1/q_n^2$ for $\lb\in \mathcal B_{\rm iso}$). The spectral cover we seek will be a cover of the union of spectra of these blocks with added rank 2
small-norm perturbations over all such perturbations and over $x$. To apply the method of \cite{bjk}, we need to find a matrix $K$ such that $\mathcal E:=\frac{d}{dx}B(x)-[K(x),B(x)]$
has a bounded rank and norm. We find two such matrices: one for 
$\lb\in \mathcal L_{I}\cup \mathcal L_{II}\cup \mathcal L_{III}$, and another
for $\lb\in\mathcal I_{\rm iso}$. The required spectral covers of at most $q_n+q_{n-1}$ intervals
of total length $O(1/q_n)$ (improving to $O(1/q_n^2)$ for $\lb\in \mathcal B_{\rm iso}$)
then follow.
In Section \ref{sec:dimensions} we use these covers to prove the theorems stated above.

Similar results on fractal dimensions hold for the quantum graph graphene \cite{bjk3}.

}

{{\it{Acknowledgement and AI statement:}} 
The research of SJ was supported by NSF DMS-2052899, DMS2155211, and Simons 896624.}{We used AI to accelerate determination of the explicit formulae for matrices $K$,
and to help find the transform of Lemma \ref{lem:LII-gauge}. 
These constructions are ideologically close to those of \cite{bjk}.
We also used AI to assist in preparing an initial draft of this manuscript, which the authors subsequently rewrote in full. }

\section{Spectral covers for irrational flux}\label{sec:covers}

We  use the direct-integral notation
\[
 M_{\alpha,\lambda}=\int_\T^\oplus H_{\alpha,\theta,\lambda}\dd\theta,
 \qquad \Sigma_{\alpha,\lambda}=\spec(M_{\alpha,\lambda}).
\]
For irrational $\alpha$, $\Sigma_{\alpha,\lambda}$ coincides with the spectrum of every fiber \cite{as}.

In this section we will prove the following: 
\begin{theorem}\label{thm:all-scales}
Let $\al$ be irrational with canonical approximants $p_n/q_n$.
For $\lambda\in\mathcal C_0$ and every $n\geq4$, there is a
cover of $\Sigma_{\alpha,\lambda}$ by closed intervals $I_{n,j}$ such that
\begin{equation}\label{ex:eq:allscale-cover}
\Sigma_{\alpha,\lambda}\subset\bigcup_{j=1}^{M_n}I_{n,j},\qquad
M_n\leq q_n+q_{n-1},\qquad
\sum_j|I_{n,j}|\leq\frac{C\max\{1,\lambda_1+\lambda_3\}}{q_n}.
\end{equation}

If $\lb\in\mathcal L_{I},\mathcal L_{III}$,
a cover with the same interval count satisfies
\begin{equation}\label{eq:boundary-refined-cover}
\sum_j|I_{n,j}|\leq C\left(\frac{|\lb_1-\lb_3|}{q_n}+\frac{\lb_1+\lb_3}{q_n^2}\right).
\end{equation}
If $\lambda=(\lb_1,1,\lb_3)\in\mathcal L_{\mathrm {II}}$, the same interval count holds and
\begin{equation}\label{eq:LII-refined-cover}
 \sum_j|I_{n,j}|\leq C\left(\frac{\sqrt{1-4\lb_1\lb_3}}{q_n}+\frac1{q_n^2}\right).
\end{equation}
The constant $C>0$ above is absolute.
\end{theorem}

{
The Hausdorff-measure bounds of Theorem~\ref{thm:HD} now follow immediately from
H\"older's inequality. Indeed, if $M_n=O(q_n)$ and
$\sum_j|I_{n,j}|=O(q_n^{-\varkappa})$, then
\[
 \sum_j|I_{n,j}|^{1/(1+\varkappa)}
 \le M_n^{\varkappa/(1+\varkappa)}
       \left(\sum_j|I_{n,j}|\right)^{1/(1+\varkappa)}=O(1).
\]
The largest interval length tends to zero, so $\varkappa=1,2$ give finite
Hausdorff measure at exponents $1/2,1/3$, respectively. The box-dimension
bounds follow from Proposition~\ref{prop:dimension} below.
\par}

%

{
For the square and triangular lattice models we use the notation
\par}
\begin{equation}\label{eq:special-couplings}
 \lambda_{\mathrm A}=(0,1,0),\qquad \lambda_{\mathrm T}=(1,1,0),
 \qquad
 \Sigma_\alpha^{\mathrm A}=\Sigma_{\alpha,\lambda_{\mathrm A}},\quad
 \Sigma_\alpha^{\mathrm T}=\Sigma_{\alpha,\lambda_{\mathrm T}}.
\end{equation}
{
In particular, $H_{\alpha,\theta,\lambda_{\mathrm A}}$ is the critical almost
Mathieu operator, and $M_{\alpha,\lambda_{\mathrm A}}=M_\alpha$ in the notation of
\cite{bjk}.
\par}

\subsection{Reduction to singular Jacobi families}\label{sec:singular-structure}
\leavevmode{
Let $J$ be the bounded Jacobi operator
\[
 (Ju)_j=\overline{b_{j-1}}u_{j-1}+v_ju_j+b_ju_{j+1},
 \qquad v_j\in\R,\quad b_j\in\C.
\]
For continuous sampling functions $v:\T\to\R$ and $b:\T\to\C$, we set
$v_j=v(\theta+j\alpha)$ and $b_j=b(\theta+j\alpha)$. When $b$ is real, 
$J=H_{v,b,\alpha,\theta}$ of \cite{bjk}. We call the family singular if $b$ has
a zero. For irrational $\alpha$, singularity gives
\[
 \inf_{j\in\Z}|b(\theta+j\alpha)|=0
\]
for every $\theta$. Deleting the bond between $j$ and $j+1$ changes the
operator by
\[
 J-J^{\rm cut}=b_j e_je_{j+1}^*+\overline b_j e_{j+1}e_j^*,
 \qquad \norm{J-J^{\rm cut}}=|b_j|,
\]
where $e_j$ is the standard basis vector. As in \cite{bjk}, we will use
continued-fraction returns to choose cuts where these coefficients are small,
and construct covers for the spectrum of perturbed blocks.
{\mg Thus, we first need to see if the family (in $\th$) of $H_{\al,\th,\lb}$ can be reduced to a family of singular Jacobi matrices.}

As we discuss below the families for $\lb\in\mathcal L_{III},\mathcal I_{\mathrm{iso}}$ are already singular. We will now reduce the families for
$\lb\in\mathcal L_{I},\mathcal L_{II}$ to singular form by using the duality
transform for $\lb\in\mathcal L_{I}$, and a chiral-type transform for $\lb\in\mathcal L_{II}$.

\par}

Aubry duality exchanges the potential amplitude and the constant hopping.
We therefore introduce a potential amplitude $\nu\ge0$ in \eqref{eq:c}:

\begin{equation}\label{extH}
 \begin{split}
 (H_{\nu,\alpha,\theta,\lambda}u)_j
 ={}&\overline{c_{\lambda,\alpha}(\theta+(j-1)\alpha)}u_{j-1}
     +c_{\lambda,\alpha}(\theta+j\alpha)u_{j+1}\\
 &+2\nu\cos(2\pi(\theta+j\alpha))u_j.
 \end{split}
\end{equation}

We write
\[
 \mathcal S_{\nu,\lambda}(\alpha)
 =\bigcup_{\theta\in\T}\spec(H_{\nu,\alpha,\theta,\lambda}).
\]
By norm continuity in $\theta$, this is also the spectrum of the direct
integral. For irrational $\alpha$, application of the unitary transform 
$(Tu)_j=u_{j+1}$ shows that it is the spectrum of every fiber. Note that
$\mathcal S_{1,\lambda}(\alpha)=\Sigma_{\alpha,\lambda}$.

We first recall the usual Aubry duality, allowing also $\lambda_2=0$.

\begin{lemma}\label{lem:lattice-duality}
For every $\alpha\in\R$, $\nu\ge0$, and $\lambda=(\lb_1,\lb_2,\lb_3)\in[0,\infty)^3$, the
families $H_{\nu,\alpha,\theta,\lambda}$ and
$H_{\lambda_2,\alpha,\theta,(\lambda_3,\nu,\lambda_1)}$ have unitarily equivalent
direct integrals. In particular,
\begin{equation}\label{ex:eq:changes}
 \mathcal S_{\nu,(\lambda_1,\lambda_2,\lambda_3)}(\alpha)
 =\mathcal S_{\lambda_2,(\lambda_3,\nu,\lambda_1)}(\alpha).
\end{equation}
For $\nu=1$ and $\lambda_2>0$, setting
$\lambda^\vee=(\lambda_3/\lb_2,1/\lb_2,\lambda_1/\lb_2)$ gives
$\Sigma_{\alpha,\lambda}=\lambda_2\Sigma_{\alpha,\lambda^\vee}$.
\end{lemma}

\begin{proof}
This is the same Aubry duality calculation as in \cite[Theorem~3.1]{GJLS}.
On $L^2(\T;\ell^2(\Z))$, let
\[
 (T\varphi)(\theta,j)=\varphi(\theta,j+1),\qquad
 (S\varphi)(\theta,j)=e^{2\pi i(\theta+j\alpha)}\varphi(\theta,j).
\]
These operators satisfy the commutation relation $TS=e^{2\pi i\alpha}ST$.
Applying $S^{-1}$ to it from both sides yields $S^{-1}T=e^{2\pi i\alpha}TS^{-1}$,
and similarly applying $T^{-1}$ we have $ST^{-1}=e^{2\pi i\alpha}T^{-1}S$.

The Fourier transform
\[
 (U\varphi)(\theta,j)=\sum_{k\in\Z}e^{-2\pi ik(\theta+j\alpha)}
        \int_\T e^{-2\pi ijx}\varphi(x,k)\dd x
\]
is unitary, with the sum interpreted in $L^2$, and satisfies
$US=T^{-1}U$ and $UT=SU$ (and consequently, by taking inverses and applying $U$,  $US^{-1}=TU$,
$UT^{-1}=S^{-1}U$) exactly as in that proof. 
We now write \eqref{extH} in terms of $S$ and $T$,

We have
\[
\begin{aligned}
 H_{\nu,\al,\th,\lambda}={}&\nu(S+S^{-1})+\lambda_2(T+T^{-1})\\
 &+\lambda_1e^{-\pi i\alpha}(S^{-1}T+ST^{-1})
  +\lambda_3e^{\pi i\alpha}(ST+S^{-1}T^{-1}).
\end{aligned}
\]
Let $Q_{\nu,\lambda}=\int_\T^\oplus H_{\nu,\alpha,\theta,\lambda}\dd\theta$.
The above commutation relations then give
\[
 UQ_{\nu,\lambda}=Q_{\lambda_2,(\lambda_3,\nu,\lambda_1)}U,
\]
which proves the claim. The relation $\Sigma_{\alpha,\lambda}=\lambda_2\Sigma_{\alpha,\lambda^\vee}$ for $\nu=1$, $\lb_2>0$, follows by factoring out $\lambda_2$.
\end{proof}


As we discuss below, if $\lb\in\mathcal L_{I}$, the family for $\lambda^\vee$
is singular.

The remaining step is to make $H$ for $\lb$ in the self-dual triangle (see Figure \ref{fig:parameter-regions}) of $\mathcal L_{II}$ singular. We use a
diagonal phase gauge, which allows to replace a hopping by any other hopping with
the same modulus; see \cite[Section~7, footnote~9]{AJM} and
\cite[(1.57) and Lemma~1.6]{Teschl}. The identity below is a simple instance
of reflecting a polynomial factor across the unit circle
\cite[Section~3]{BeinertPlonka}. Combined with Aubry duality, it gives the
desired singular representative.

\begin{lemma}\label{lem:LII-gauge}
Let $\lambda=(a,1,b)\in\mathcal L_{\mathrm{II}}$, and set
\be
 \Lambda=a+b,\qquad r_\pm=\frac{1\pm\sqrt{1-4ab}}2.
\ee
For every $\alpha,\theta$, there is a diagonal unitary
$G_{\alpha,\theta}$ such that
\begin{equation}\label{eq:LII-gauge}
 G_{\alpha,\theta}^*H_{\alpha,\theta,(a,1,b)}G_{\alpha,\theta}
 =H_{\alpha,\theta,(r_+,\Lambda,r_-)}.
\end{equation}
\end{lemma}

\begin{proof}
Since $r_++r_-=1$ and $r_+r_-=ab$, for $|z|=1$,
\[
\begin{aligned}
 |az^{-1}+1+bz|^2
 &=1+a^2+b^2+2\Lambda\operatorname{Re}z
       +2ab\operatorname{Re}(z^2)\\
 &=|r_+z^{-1}+\Lambda+r_-z|^2.
\end{aligned}
\]
Therefore, the hopping moduli coincide:
$|c_{\lb,\al}(\th+j\al)|=|c_{(r_+,\Lambda,r_-),\al}(\th+j\al)|$.
If this modulus is nonzero, denote $c_j^{(1)}=\exp\{i\arg c_{\lb,\al}(\th+j\al)\}$, $c_j^{(2)}=\exp\{i\arg c_{(r_+,\Lambda,r_-),\al}(\th+j\al)\}$,
otherwise set $c_j^{(1)}=c_j^{(2)}=1$.
Let
\be
 g_{j+1}=
 g_jc_j^{(2)}/c_j^{(1)},\qquad j\in\Z;\qquad g_0=1.
\ee
Then $(G_{\alpha,\theta}u)_j=g_ju_j$ gives
\eqref{eq:LII-gauge}. The gauges are measurable in $\theta$,
so they also give unitary equivalence of the direct integrals.
\end{proof}

\begin{theorem}\label{thm:LII-singular}
With the notation of \eqref{extH} and Lemma \ref{lem:LII-gauge},
for every $\alpha\in\R$ and $\lambda\in\mathcal L_{\mathrm{II}}$,
$M_{\alpha,\lambda}$ is unitarily equivalent to the direct integral
of a singular Jacobi family $H_{\Lambda,\alpha,\theta,(r_-,1,r_+)}$.
\end{theorem}

\begin{proof}
Apply Lemma~\ref{lem:LII-gauge}, followed by the Aubry duality
(Lemma~\ref{lem:lattice-duality}). The resulting family is
$H_{\Lambda,\alpha,\theta,(r_-,1,r_+)}$, whose hopping vanishes
at $t=1/2-\alpha/2$, since $r_-+r_+=1$.
\end{proof}

\begin{corollary}\label{cor:singular-representatives}
For every $\alpha\in\R$ and $\lambda\in\mathcal C_0$,
$M_{\alpha,\lambda}$ is unitarily equivalent to the direct integral
of the following singular Jacobi family at frequency $\alpha$ with $b(0)=0$:
\be\label{sing}
(H_{v,b,\alpha,\theta}\phi)(n)=b(\th + (n-1)\al)\phi(n-1)+b(\th+n\al)
\phi(n+1)+v(\th+n\al)\phi(n),
\ee
where for $\lb\in \mathcal L_{I}, \mathcal L_{II},\mathcal L_{III}$,
\be\label{vb}
v(x)=-2\nu \cos(2\pi x-\pi\al),\qquad b(x)=2w(\pi x) \sin(\pi x),\quad w(t)=\sqrt{s^2\sin^2 t + d^2 \cos^2 t},
\ee
with parameters 
\be\label{param}\begin{cases}
\nu=\lb_2, s=\lb_1+\lb_3=1, d=\lb_3-\lb_1 & \mbox{for $\mathcal L_{\mathrm I}$},\\ 
\nu=\lb_1+\lb_3, s=1, d=-\sqrt{1-4\lb_1\lb_3} & \mbox{for $\mathcal L_{\mathrm{II}}$},\\
\nu=1, s=\lb_1+\lb_3=\lb_2, d=\lb_1-\lb_3 & \mbox{for $\mathcal L_{\mathrm{III}}$.}
\end{cases}
\ee
For $\lb\in \mathcal I_{\mathrm{iso}}$,
\be\label{vb-iso}
v(x)=-2\cos(2\pi x-\pi\al+\ga),\qquad b(x)=\lb_2-2\lb_1\cos(2\pi x +\ga),\qquad
\ga=\arccos\frac{\lb_2}{2\lb_1}.
\ee
\end{corollary}

\begin{proof}
On the three critical boundaries, choose
\begin{equation}\label{ex:eq:representatives}
\begin{aligned}
 \lambda\in\mathcal L_{\mathrm{III}}:
 &\quad \nu=1,\quad \widetilde\lambda=\lambda,\\
 \lambda\in\mathcal L_{\mathrm I}:
 &\quad \nu=\lambda_2,\quad
       \widetilde\lambda=(\lambda_3,1,\lambda_1),\\
 \lambda=(a,1,b)\in\mathcal L_{\mathrm{II}}:
 &\quad \nu=a+b,\quad
       \widetilde\lambda=(r_-,1,r_+),
\end{aligned}
\end{equation}
where $r_\pm$ are defined in Lemma~\ref{lem:LII-gauge}.
These are, respectively, the original family, its Aubry dual,
and the family constructed in Theorem~\ref{thm:LII-singular}.
In each case,
\be\label{lambdanew}
 \mathcal S_{1,\lambda}(\alpha)
 =\mathcal S_{\nu,\widetilde\lambda}(\alpha),\qquad
 0\le\nu\le\widetilde\lambda_2
 =\widetilde\lambda_1+\widetilde\lambda_3>0.
\ee
Writing $z=e^{2\pi i(t+\alpha/2)}$, we have $c_{\widetilde\lambda,\alpha}(t)
 =z^{-1}(1+z)(\widetilde\lambda_1+\widetilde\lambda_3z),
$ so the hopping vanishes at $t=1/2-\alpha/2$.
We further denote
\[
s=\wt\lb_1+\wt\lb_3,\qquad d=\wt\lb_1-\wt\lb_3,
\]
and write
\be
\wh b(t):=c_{\widetilde\lambda,\alpha}(t-\al/2-1/2)=
2w(\pi t)\sin(\pi t)\ga(\pi t),\qquad \ga(t):=e^{i\arg(s\sin t+i d\cos t)}.
\ee
{\mg At zeros of $s\sin t + id \cos t$, set $\ga(t) = 1$.}
Now similarly to Lemma \ref{lem:LII-gauge}, let
\[
 g'_{j+1}=
 g'_j/\ga(\pi(\th+\al j)),\qquad j\in\Z;\qquad g'_0=1.
\]
Then $(G'_{\alpha,\theta}u)_j=g'_ju_j$ is the unitary transform of $H$ which
removes the factors $\ga(\pi t)$ from the off-diagonal matrix elements without changing diagonal elements and thus reduces
the family to the form \eqref{sing} as the same phase shift 
$t\to t-\al/2-1/2$ has to be applied to the diagonal elements as well.

For $\lb\in\mathcal I_{\mathrm{iso}}$, the original hopping
$c_{\lambda,\alpha}(t)
=\lambda_2+2\lambda_1\cos(2\pi(t+\alpha/2))$
already has a zero, since $\lambda_2\le 2\lambda_1$.
Choosing $b(t)=c_{\lambda,\alpha}(t-\al/2+1/2+\ga/(2\pi))$, 
$\ga=\arccos\frac{\lb_2}{2\lb_1}$, we complete the proof.
\end{proof}

\subsection{Separation into blocks}\label{sec:separation}

Consider the operator $H_{v,b,\al,\th}$ given by \eqref{sing}. Note that it is singular with $b(0)=0$. 
We now construct covers of its spectrum applying
the strategy of our article \cite{bjk}. First, we separate it into blocks and use a submersion lemma.

Let $p_n/q_n$ be a sequence of canonical approximants to $\al$ and let
\be\label{Jn}
\de_n=q_n\left(\alpha-\frac{p_n}{q_n}\right),\quad J_n=[-|\de_{n-1}|,|\de_{n-1}|]\qquad n=2,3,\dots.
\ee
By standard properties of approximants, the length
\[
|J_n|<2/q_n.
\]

Fix any $n\ge 3$ and any phase $\th=\th_0$ of $H_{v,b,\al,\th}$. 
 Choose
$n_0\in\mathbb Z$ such that $x_0 :=\th_0+n_0\al$ satisfies
$x_0\pmod{1}\in(-|\de_n|,|\de_n|)$.
{ There is a strictly increasing sequence of integers $m_k$, $k=\dots,-1,0,1,\dots$,
such that $m_0=0$, $N_k:=m_{k+1}-m_k$ takes on only two values, $N_k=q_n$ or $N_k=q_{n-1}$,  and such that
$x_0+\al m_k\pmod{1}\in J_n$, $k\in\Z$, and $x_0+\al j\pmod{1}\notin (-|\de_n|,|\de_n|)$ if $j\neq m_k$ for any $k$ (this is discussed in more detail in  Section 7.1. of \cite{bjk})}.

As in \cite{bjk}, 
we separate $H_{v,b,\al,\th}$ into blocks removing the matrix elements $b(x_k)$, where $x_k:=\th_0+(n_0+m_k)\alpha$. Then $x_k\pmod{1}\in J_n$ and $|b(x_k)|\le C(b)|J_n|$.
The length of the $k$'th block between the removed bonds $b(x_k)$ and $b(x_{k+1})$ is $N_k=m_{k+1}-m_k\in\{q_n,q_{n-1}\}$.

 Denote an orthonormal basis of $\C^{N}$ by $e_0,e_1,\dots,e_{N-1}$,
 and write a vector in this basis as $u=(u_0,u_1,\dots, u_{N-1})^T$.
 Let $\wh \Gamma_N: \C^{N}\to \C^{2}$ and the projection $P_N:  \C^{N}\to \C^{N}$ onto  the $\operatorname{span}\{e_0,e_{N-1}\}$,
 be given by
\be\label{PN}
\wh\Gamma_Nu=(u_0,u_{N-1})^T,\qquad
P_N=\wh\Gamma_N^*\wh\Gamma_N.
\ee

 Define a tridiagonal matrix $\wh B_N(x):\C^{N}\to\C^{N}$ using the matrix elements of  $H_{v,b,\al,\th}$ as
 \begin{equation}\label{BhN}
\begin{aligned}
\wh B_N(x)_{j,j}&=v(x+(j+1)\alpha), \qquad 0\le j\le N-1,\cr
 \wh B_N(x)_{j+1,j}&=\wh B_N(x)_{j,j+1}=b(x+(j+1)\alpha),
 \qquad 0\le j\le N-2.
 \end{aligned}
\end{equation}


The following proposition from \cite{bjk} applies here.

\begin{proposition}\cite[Proposition~7.3]{bjk}
\label{prop:blocks}
Let $H_{v,b,\alpha,\th}$ be given by \eqref{sing} with { $b,v\in C^1(\mathbb R;\mathbb R)$}, $v(x)$ periodic of period $1$,
$b(x)$ either periodic of period $1$ with $b(0)=0$, or of period 2 and $b(0)=b(1)=0$. Let $\al$ be irrational, $\th\in\R$, $\wh\Gamma_N$ given in \eqref{PN},
 $\wh B_N(x)$ given in \eqref{BhN}, and $J_n$ in \eqref{Jn}. Then for $N'_1=q_n$ and $N'_2=q_{n-1}$, $n\ge 3$, $\tau=\sup_k|b(x_k)|$,
\begin{align}
 \spec(H_{v,b,\alpha,\th})\subset
 \overline{\bigcup_{\nu=1}^{2}
 { \ \bigcup_{x\in J_n}}
 \ \bigcup_{\substack{V=V^*\in \mathcal M_2(\C)\\\|V\|\le\tau}}
 {\si\left(\wh B_{N'_\nu}(x)+\wh\Gamma_{N'_\nu}^*V\wh\Gamma_{N'_\nu}\right)}}.
\end{align}
\end{proposition}

Note that $\wh\Gamma_{N'_\nu}^*V\wh\Gamma_{N'_\nu}$ has only 4 nonzero matrix elements located at upper right, left,
and lower right, left corners of the matrix. Thus the spectrum of $H_{v,b,\alpha,\th}$ is included into the union of spectra of perturbed blocks. The perturbation is of rank at most 2 and of small (for large $n$) norm 
\[
\| V\|\le\tau=\sup_k|b(x_k)|\le C(b)|J_n|<2C(b)/q_n.
\]
We will see that in some cases, the bound on $\sup_k|b(x_k)|$ is smaller.

\subsection{General cover estimates}\label{sec:variation}
{\mg A proof of the following proposition is contained in Section 7.2. of \cite{bjk} in the case of constant $K$. We reproduce it here briefly for reader's convenience  and extend to $K$ depending on the phase $x$.} In the following sections we find the matrices $K$, $\mathcal E$ for $\lb\in\mathcal C_0$ 
with $\mathcal E$ of rank no more than 4 satisfying $(\Id- P_N)\mathcal E (\Id-P_N)=0$,
and derive explicit expressions for the estimates in Proposition \ref{prop:block-cover} below for $\lb\in\mathcal C_0$.

We write $\|A\|_{S_1}=\Tr\sqrt{A^*A}$, and $\|A\|_{S_2}^2=\Tr(A^*A)$ for trace and Hilbert--Schmidt norms of a matrix $A$, respectively.

\begin{proposition}\label{prop:block-cover}
Let $N\ge 2$, $J$ be a compact interval, let $B:J\to\mathcal M_N(\C)$ be a $C^1$ Hermitian path, and let $K:J\to\mathcal M_N(\C)$ be continuous with $K^*=-K$. Let
\[
\mathcal E(x)=
 \frac{d}{dx}B(x)-[K(x),B(x)].
\]
Suppose that $(\Id- P_N)\mathcal E (\Id-P_N)=0$, where $P_N$ is defined in \eqref{PN}.
For $\bt\ge0$, let $H^\pm=B\pm \bt P_N$ and define the closed intervals
\[
 I_j=\bigl[\min_{x\in J}\lambda_j(H^-(x)),\ \max_{x\in J}\lambda_j(H^+(x))\bigr],
 \qquad 1\le j\le N,
\]
where $\lb_1(H^{\pm})\le\cdots\le \lb_N(H^{\pm})$ are the ordered eigenvalues
of $H^\pm$. Then
\be\label{incl}
 \bigcup_{x\in J}\ \bigcup_{\substack{V=V^*\in\mathcal M_2(\C)\\\|V\|\le \bt}}
 \spec\bigl(B(x)+\widehat\Gamma_N^*V\widehat\Gamma_N\bigr)
 \subset\bigcup_{j=1}^N I_j,
\ee
and the total length of intervals satisfies
\begin{equation}\label{eq:boundary-length}
 \sum_{j=1}^N|I_j|\le4\bt+4\int_J\bigl(\|\mathcal E(x)\|_{S_2}+\bt\|[K(x),P_N]\|_{S_2}\bigr)\dd x.
\end{equation}
\end{proposition}

\begin{proof}
For $V=V^*$ with $\|V\|\le \bt$,
\[
 -\bt P_N\le\widehat\Gamma_N^*V\widehat\Gamma_N\le \bt P_N,
 \qquad H^-\le B+\widehat\Gamma_N^*V\widehat\Gamma_N\le H^+.
\]
The min--max principle gives, for every $x\in J$,
\[
 \lambda_j(H^-(x))\le\lambda_j(B(x)+\widehat\Gamma_N^*V\widehat\Gamma_N)
 \le\lambda_j(H^+(x)),
\]
which proves the inclusion \eqref{incl}.

Now fix $x_0\in J$. Note that
\[
 \sum_{j=1}^N\bigl(\lambda_j(H^+(x_0))-\lambda_j(H^-(x_0))\bigr)
 =\Tr(2\bt P_N)=4\bt.
\]
Adding the variations for each $j$ gives
\begin{equation}\label{eq:boundary-width}
 \sum_{j=1}^N|I_j|\le4\bt +\sum_{j=1}^N\Var_J\lambda_j(H^+)
 +\sum_{j=1}^N\Var_J\lambda_j(H^-).
\end{equation}
Note that $(\Id- P_N) [K,P_N] (\Id-P_N)=0$, and therefore in particular
$\rank [K,P_N]\le 4$ and so $\|[K,P_N]\|_{S_1}\le 2 \|[K,P_N]\|_{S_2}$.
Let $\mathcal E^\pm=\frac{d}{dx}H^\pm - [K,H^\pm]$.
Then we clearly have using our assumption on $\mathcal E$ that
\[
 \mathcal E^\pm=\mathcal E\mp \bt [K,P_N],\qquad
 (I-P_N)\mathcal E^\pm(I-P_N)=0,
\]
and so $\rank \mathcal E^\pm\le 4$. Thus 
$\|\mathcal E^\pm\|_{S_1}\le 2  \|\mathcal E^\pm\|_{S_2}$.

{\mg
Lemma~7.5 in \cite{bjk} applied to $H^+$ and then to $H^-$ would 
bound the sum of variations by integrals of $\|\mathcal E^\pm\|_{S_1}$. However, it was proved
in \cite{bjk} for a constant $K$. The proof is based on the unitary transform $e^{-(x-a)K}$.
To extend the proof to the present case of $K(x)$ depending on $x$, consider instead the unitary transform
$U(x)$ which is the solution to the equation $\frac{d}{dx}U(x)=-U(x)K(x)$ with initial condition $U(a)=I$.
(The fact that it is unitary is easy to see differentiating $U(x)U(x)^*$.) Otherwise, the proof remains the same. Then, the lemma yields}
\[
\begin{aligned}
 \sum_{j=1}^N\bigl(\Var_J\lambda_j(H^+)+\Var_J\lambda_j(H^-)\bigr)
 &\le\int_J\bigl(\|\mathcal E^+\|_{S_1}+\|\mathcal E^-\|_{S_1}\bigr)\dd x\le4\int_J\bigl(\|\mathcal E\|_{S_2}+\bt\|[K,P_N]\|_{S_2}\bigr)\dd x.
\end{aligned}
\]
Together with \eqref{eq:boundary-width}, this proves \eqref{eq:boundary-length}.
\end{proof}

\subsection{Spectral covers on the three critical boundaries}\label{sec:cover3b}

{\mg We now provide explicit matrices $K$ and $\mathcal E$ for 
$\lb\in {\mathcal L}_{I},{\mathcal L}_{II},{\mathcal L}_{III}$,
such that  $\mathcal E$ and $\beta [K,P_N]$ have bounded Hilbert-Schmidt norm and $(\Id- P_N)\mathcal E (\Id-P_N)=0$. 
Applying then Proposition \ref{prop:block-cover} with $J=J_n$ and 
$\bt=C(2s/q_n^2+|d|/q_n)$, we obtain smallness of the total length of the cover.
In the next subsection we do the same for $\lb\in {\mathcal I}_{iso}$.}

Recall  \eqref{vb}, \eqref{param}. Assume first that $d\neq 0$. We then have $0<|d|\le w(t)\le s$ for all $t\in\R$.

{\mg
Let $\rho=\sqrt{s^2-\nu^2}$ and
\[
 C(t)=\begin{pmatrix}
       \nu\sin t -i d\cos t&\rho\sin t\\
       -\rho\sin t&\nu\sin t +i d\cos t
      \end{pmatrix},\qquad M(t)=\frac{C(\pi t)}{w(\pi t)}.
\]
For $r\ge1$, we write
\[
 \mathcal P_r(t)=M(t)M(t+\al)\cdots M(t+(r-1)\al),\qquad
 \chi_r(t)=\tfrac12\Tr\mathcal P_r(t).
\]
Denote
\be\label{tj}
 t_j=x+(j+1)\alpha,\qquad j\in\Z.
\ee 
We define the coefficients $k_r$ and the entries of $K_N(x)$ by the formulae
\begin{equation}\label{K}
 \begin{split}
 k_r(t)&=\frac{\pi(-1)^{r-1}\chi_r(t)}{\sin(\pi\al r)},\qquad r\ge1,\\
 (K_N(x))_{j,j+r}&=k_r(t_j),\qquad
 (K_N(x))_{j+r,j}=-k_r(t_j),\qquad 0\le j<j+r\le N-1,\\
 (K_N(x))_{jj}&=0,\qquad 0\le j\le N-1.
 \end{split}
\end{equation}}
Note that $\sin(\pi\al r)\neq 0$ since $\al$ is irrational.

\begin{lemma}\label{lemmaEK}
Let $N>2$, $B_N(x)=\wh B_N(x)$ be given by \eqref{BhN} with $v(x)$, $b(x)$ as defined in \eqref{vb}, and 
let $K_N(x)$ be as defined in \eqref{K}.
Then 
\[
 \mathcal E_N(x):=\frac{d}{dx}B_N(x)-[K_N(x),B_N(x)]
 \]
has the following edge matrix elements for $1\le j\le N-2$,
\be\label{Eedge}\begin{aligned}
(\mathcal E_N)_{0,j}&=(\mathcal E_N)_{j,0}=\frac{(-1)^{j+1}2\pi}{\sin(\pi(j+1)\al)}w(\pi t_{-1})\sin (\pi t_{-1})\chi_{j+1}(t_{-1}),\\
(\mathcal E_N)_{j,N-1}&=(\mathcal E_N)_{N-1,j}=\frac{(-1)^{N-1-j}2\pi}{\sin(\pi(N-j)\al)}w(\pi t_{N-1})\sin (\pi t_{N-1})\chi_{N-j}(t_j),\\
\end{aligned}
\ee
and the corner elements
\be\label{Ecorner}\begin{aligned}
(\mathcal E_N)_{0,0} &=- \frac{4\pi}{\sin(\pi\al)}w(\pi t_{-1})\sin (\pi t_{-1})\chi_{1}(t_{-1}),\\
(\mathcal E_N)_{N-1,N-1}&=\frac{4\pi}{\sin(\pi\al)}w(\pi t_{N-1})\sin (\pi t_{N-1})\chi_{1}(t_{N-1}),\\
(\mathcal E_N)_{0,N-1}&= (\mathcal E_N)_{N-1,0}=\pi (-1)^{N-1} \Tr\mathcal P_{N-1}(t_0)
\begin{pmatrix}
\nu\sin\ga-id\cos\ga & \rho \sin\ga\\
-  \rho \sin\ga & \nu\sin\ga +id\cos\ga
\end{pmatrix},\\ \ga &=\pi (t_{-1}+t_{N-1}),
\end{aligned}
\ee
and the rest of the matrix elements of $\mathcal E_N(x)$ are zero.
\end{lemma}

\begin{proof}
We omit the subscript $N$. The identities $B^T=B$ and $K^T=-K$ give $[K,B]^T=[K,B]$, so it suffices to compute the entries of $\mathcal E$ on and above the diagonal. 

First, let $1\le j\le N-2$, $\ell\ge 2$, $j+\ell\le N-2$. Then using the fact that $B$ is tridiagonal, we have
\be\label{commut}
\begin{aligned}
[K,B]_{j,j+\ell}&= \sum_k K_{j,k}B_{k,j+\ell}- B_{j,k}K_{k,j+\ell}\\
&= K_{j,j+\ell}B_{j+\ell,j+\ell}+K_{j,j+\ell-1}B_{j+\ell-1,j+\ell}+K_{j,j+\ell+1}B_{j+\ell+1,j+\ell}\\
&- K_{j,j+\ell}B_{j,j}-K_{j-1,j+\ell}B_{j,j-1}-K_{j+1,j+\ell}B_{j,j+1}\\
&=k_{\ell}(t_j)v(t_{j+\ell})+k_{\ell-1}(t_j)b(t_{j+\ell-1})+k_{\ell+1}(t_j)b(t_{j+\ell})\\
&- k_{\ell}(t_j)v(t_{j})-k_{\ell+1}(t_{j-1})b(t_{j-1})-k_{\ell-1}(t_{j+1})b(t_{j}).
\end{aligned}
\ee
Now denoting
\[
A= \mathcal P_{\ell}(t_j),\qquad s_j=\sin(\pi\al j),
\]
we remark that
\[
\begin{aligned}
 k_\ell(t_j)&=\frac{(-1)^{\ell-1}\pi}{2 s_{\ell}}\Tr A,\\
 k_{\ell-1}(t_j)&=\frac{(-1)^{\ell}\pi}{2 s_{\ell-1}}\Tr AM(t_{j+\ell-1})^{-1},\qquad
 k_{\ell-1}(t_{j+1})=\frac{(-1)^{\ell}\pi}{2 s_{\ell-1}}\Tr AM(t_{j})^{-1},\\
 k_{\ell+1}(t_j)&=\frac{(-1)^{\ell}\pi}{2 s_{\ell+1}}\Tr AM(t_{j+\ell}),\qquad
 k_{\ell+1}(t_{j-1})=\frac{(-1)^{\ell}\pi}{2 s_{\ell+1}}\Tr AM(t_{j-1}),
 \end{aligned}
 \]
and therefore
\be\label{mid}\begin{aligned}
[K,B]_{j,j+\ell}&=\frac{(-1)^{\ell}\pi}{2s_{\ell}s_{\ell-1}s_{\ell+1}}
\Tr A\left[ (v(t_j)-v(t_{j+\ell}))s_{\ell+1}s_{\ell-1}\Id\right.\\
&\left. +(M(t_{j+\ell})b(t_{j+\ell})-M(t_{j-1})b(t_{j-1}))s_{\ell}s_{\ell-1}\right.\\
&\left.+
(M(t_{j+\ell-1})^{-1}b(t_{j+\ell-1})-M(t_{j})^{-1}b(t_{j}))s_{\ell}s_{\ell+1}\right].
\end{aligned}
\ee
Using trigonometric identities, it is straightforward to verify that with $\ga_j=\pi(t_{j-1}+t_{j+\ell})$
\be
(M(t_{j+\ell})b(t_{j+\ell})-M(t_{j-1})b(t_{j-1}))s_{\ell}s_{\ell-1}=
2s_{\ell}s_{\ell+1}s_{\ell-1}
\begin{pmatrix}
\nu\sin\ga_j- id\cos\ga_j & \rho\sin\ga_j\\
-\rho\sin\ga_j& \nu\sin\ga_j+ id\cos\ga_j
\end{pmatrix},
\ee
and since $\ga_j=\pi(t_{j-1}+t_{j+\ell})=\pi(t_{j}+t_{j+\ell-1})$,
\be
(M(t_{j+\ell-1})^{-1}b(t_{j+\ell-1})-M(t_{j})^{-1}b(t_{j}))s_{\ell}s_{\ell+1}=
2s_{\ell}s_{\ell+1}s_{\ell-1}
\begin{pmatrix}
\nu\sin\ga_j+ id\cos\ga_j & -\rho\sin\ga_j\\
\rho\sin\ga_j& \nu\sin\ga_j- id\cos\ga_j
\end{pmatrix},
\ee
and
\[
v(t_j)-v(t_{j+\ell})=- 4\nu s_{\ell}\sin\ga_j,
\]
we obtain that the expression inside the square bracket in \eqref{mid} vanishes, so that $[K,B]_{j,j+\ell}=0$.

Using \eqref{commut} we similarly have for $1\le j\le N-3$, $\ell=1$,
\be\label{jj+1}\begin{aligned}
[K,B]_{j,j+1}&=-\frac{\pi}{2s_{1}s_{2}}
\Tr M(t_j)\left[ (v(t_j)-v(t_{j+1}))s_{2}\Id\right.\\
&\left. +(M(t_{j+1})b(t_{j+1})-M(t_{j-1})b(t_{j-1}))s_{1}\right].
\end{aligned}
\ee
Substituting the explicit expressions for $M$ and $b$ gives
\be
[K,B]_{j,j+1}=-\frac{\pi}{2w(\pi t_j)}\left(-4s^2\sin(2\pi t_j)\sin(\pi t_j)
-4d^2\cos(2\pi t_j)\cos(\pi t_j)\right)=\frac{d}{dx}b(t_j)=\frac{d}{dx}B(x)_{j,j+1}.
\ee
Now for $1\le j\le N-2$, $\ell=0$,
\be\begin{aligned}
[K,B]_{j,j}&=\frac{\pi}{s_{1}}
\Tr \left[M(t_{j})b(t_{j})-M(t_{j-1})b(t_{j-1})\right]\\
&=
4\pi\nu\sin(\pi(t_j+t_{j-1}))=\frac{d}{dx}v(t_j)=\frac{d}{dx}B(x)_{j,j}.
\end{aligned}
\ee

Consider now $1\le j\le N-2$, $j+\ell=N-1$. Then in \eqref{commut}
there is no element $K_{j,N}B_{N,N-1}$. We can define it by the same formula,
\[
K_{j,N}B_{N,N-1}:=k_{N-j}(t_j)b(t_{N-1}),\qquad N-j=\ell+1.
\]
Adding and subtracting it in \eqref{commut} we obtain
\be
[K,B]_{j,N-1}=\frac{d}{dx}B(x)_{j,N-1}-{\mathcal E}_{j,N-1},
\ee
with the element ${\mathcal E}_{j,N-1}$ as given in \eqref{Eedge}.

The other matrix elements, $[K,B]_{0,\ell}$, $[K,B]_{N-1,N-1}$ 
of the upper triangular part of $[K,B]$ are verified similarly.
\end{proof}

We also have

\begin{lemma}\label{lemma-chi}
For $d\ne0$, the matrices and coefficients above satisfy, for every real $t$ and every $r\ge1$,
\begin{equation}\label{eq:EHM-generator-bound}
 M(t),\mathcal P_r(t)\in\mathrm{SU}(2),\qquad \chi_r(t)\in[-1,1].
\end{equation}
\end{lemma}

\begin{proof}
Multiplication of matrices, using $\nu^2+\rho^2=s^2$, gives
\[
 C(t)^*C(t)=\bigl((\nu^2+\rho^2)\sin^2t+d^2\cos^2t\bigr)I
           =w(t)^2I,\qquad \det C(t)=w(t)^2\ge d^2\neq 0.
\]
Thus $M(t),\mathcal P_r(t)\in\mathrm{SU}(2)$. The eigenvalues of $\mathcal P_r(t)$ are $e^{i\varphi_r(t)}$ and $e^{-i\varphi_r(t)}$, so
\[
 \chi_r(t)=\cos\varphi_r(t)\in[-1,1].
\]
\end{proof}

\begin{lemma}\label{lemmaEKnorms}{\mg
For $n\ge4$, let $x\pmod1, x+\al N\pmod 1\in J_n$,} and 
$N\in\{q_n,q_{n-1}\}$. Then the matrices $\mathcal E_N(x)$ and $K_N(x)$ from Lemma~\ref{lemmaEK} satisfy
\begin{equation}\label{eq:EHM-bounds}
 \|\mathcal E_N\|_{S_2}\le C\left(\frac{s}{q_n}+|d|\right),\qquad
 \|[K_N,P_N]\|_{S_2}\le Cq_n,
\end{equation}
where $C>0$ is absolute.
\end{lemma}
\begin{proof}
Using Lemma \ref{lemma-chi} and the formulae for the matrix elements of $K_N(x)$ and $\mathcal E_N(x)$ 
given by \eqref{K}, \eqref{Eedge}, \eqref{Ecorner}, we obtain with an absolute $C>0$,
\be\label{Enorms}
 \begin{aligned}
 \|\mathcal E_N(x)\|_{S_2}^2
 &\le C\bigl(|w(\pi t_{-1})\sin(\pi t_{-1})|^2+|w(\pi t_{N-1})\sin(\pi t_{N-1})|^2\bigr)
       \sum_{r=1}^{N-1}\frac{1}{|\sin(\pi\al r)|^2}\cr
& + C(s^2 \sin^2\ga + d^2 ),\cr       
\|[K_N,P_N]\|_{S_2}^2
& \le C \sum_{r=1}^{N-1}\frac{1}{|\sin(\pi\al r)|^2}.
 \end{aligned}
 \ee
As argued in (87) of \cite{bjk}, $\sum_{r=1}^{L}1/|\sin(\pi\al r)|^2<Cq^2_n$,
for $L\le\lfloor q_n/2\rfloor$ with absolute $C>0$. However that argument straightforwardly extends for any $1<L<q_n$, and since $N\le q_n$, we conclude that
\be\label{sumsin}
\sum_{r=1}^{N-1}\frac{1}{|\sin(\pi\al r)|^2}< C q_n^2,
\ee
which gives the second inequality of the lemma.

{\mg Since $x\pmod1, x+\al N\pmod 1\in J_n$,} we also have that
\begin{align}\label{w1}
|w(\pi t_{-1})\sin(\pi t_{-1})|
&\le 2(s|\sin(\pi x)|+|d|)|\sin(\pi x)|\le C(s|J_n|^2 +|d||J_n|)<
C\left(\frac{s}{q_n^2}+\frac{|d|}{q_n}\right).\\
\nonumber
|w(\pi t_{N-1})\sin(\pi t_{N-1})|&\le 
2(s|\sin(\pi (x+\al N))|+|d|)|\sin(\pi (x+\al N))|\\
&<C\left(\frac{s}{q_n^2}+\frac{|d|}{q_n}\right).\label{wN}
\end{align}
{\mg Noting also that $|\sin \ga| \le |\sin(\pi t_{-1})| + |\sin(\pi t_{N-1})| \le C/q_n$ and 
substituting \eqref{sumsin}, \eqref{w1}, \eqref{wN} into \eqref{Enorms},
we obtain the first inequality of the lemma.}
\end{proof}

\begin{proposition}\label{prop:L123-cover}
Let $\lb$ be in the regions $\mathcal L_j$, $j=1,2,3$. For every irrational $\alpha$ and every $n\ge4$, there are closed intervals $I_{n,1},\ldots,I_{n,M_n}$ such that
\be\label{incl2}
 \Sigma_{\al,\lb}=\spec(H_{\al,\th,\lambda})\subset\bigcup_{j=1}^{M_n}I_{n,j},
 \qquad M_n\le q_n+q_{n-1},
\ee
and
\begin{equation}\label{L123-coverlength}
 \sum_{j=1}^{M_n}|I_{n,j}|
 \le C\left(\frac{|d|}{q_n}
                    +\frac{s}{q_n^2}\right),
\end{equation}
where $s$ and $d$ are given in \eqref{param} above, and
 $C$ is absolute.
\end{proposition}

\begin{proof}
In the singularized representation for $H_{\al,\th,\lambda}$, we
have by \eqref{lambdanew} in all 3 cases $\mathcal L_j$, $j=1,2,3$,
$s=\wt\lb_1+\wt\lb_3=\wt\lb_2>0$. 
We also have $d=\wt\lb_1-\wt\lb_3$ and 
$\wt\lb=(s/2+d/2,s,s/2-d/2)$ in all 3 cases.

We first suppose that $d\ne0$. In Proposition \ref{prop:block-cover},
set $J=J_n$, $K=K_N$, $\mathcal E=\mathcal E_N$, 
\[
\bt=\sup_k|b(x_k)|\le C(2s/q_n^2+|d|/q_n).
\] 
{\mg  Note that $B_N(x)$, $B_N(x+1)$, $B_N(x-1)$ are unitarily equivalent by the transform 
\[U=\diag\{1,-1,1,-1,\dots\},\] so we only need to consider the case $x\in J_n$.}
Using now Proposition \ref{prop:blocks} with $\tau=\bt$
and Lemma \ref{lemmaEKnorms}, we immediately obtain
\eqref{incl2} and \eqref{L123-coverlength} in this case.

Suppose now that $d=0$. Let $0<\ep<s/2$ and consider $d'=2\ep$ and the corresponding $\wt\lb=\lb^{(\ep)}=(s/2+\ep,s,s/2-\ep)$.

With $z=e^{2\pi i(x+\alpha/2)}$, the difference of the respective matrix elements given in \eqref{eq:c} is
\[
 c_{\lambda^{(\varepsilon)},\alpha}(x)-c_{\lambda^{(0)},\alpha}(x)
 =\varepsilon(z^{-1}-z),\qquad
 \|c_{\lambda^{(\varepsilon)},\alpha}-c_{\lambda^{(0)},\alpha}\|_\infty\le2\varepsilon.
\]
Hence, uniformly in phase,
\[
 \|H_{\nu,\alpha,\theta,\lambda^{(\varepsilon)}}
       -H_{\nu,\alpha,\theta,\lambda^{(0)}}\|\le4\varepsilon.
\]
Writing $\mathcal S_\varepsilon=\mathcal S_{\nu,\lambda^{(\varepsilon)}}(\alpha)$, we obtain the Hausdorff distance
\[
 d_{\rm H}(\mathcal S_\varepsilon,\mathcal S_0)\le4\varepsilon.
\]
Fix $n$ and write $M=q_n+q_{n-1}$. By \eqref{incl2} and \eqref{L123-coverlength} just proved in the case $d'=2\ep\ne0$, adding {\mg zero-length} intervals if needed, we have
\[
 \mathcal S_\varepsilon\subset U_\varepsilon
 :=\bigcup_{j=1}^M[a_{\varepsilon,j},b_{\varepsilon,j}],\qquad
 \sum_{j=1}^M(b_{\varepsilon,j}-a_{\varepsilon,j})
 \le C\left(\frac{2\varepsilon}{q_n}+\frac{s}{q_n^2}\right).
\]
Since $\mathcal S_\ep$ lies in a fixed compact interval, we may intersect the covering intervals with it and replace empty intersections by singletons in it, without increasing the total length. Thus there is a sequence $\varepsilon_k\downarrow0$ on which every endpoint $a_{\varepsilon_k,j}$, $b_{\varepsilon_k,j}$ converges. We write
\[
 a_{\varepsilon_k,j}\longrightarrow a_j,\qquad
 b_{\varepsilon_k,j}\longrightarrow b_j,\qquad
 a_j\le b_j,\qquad U_0=\bigcup_{j=1}^M[a_j,b_j].
\]
Then
\[
 d_{\rm H}(U_{\varepsilon_k},U_0)
 \le\max_{1\le j\le M}\max\{|a_{\varepsilon_k,j}-a_j|,|b_{\varepsilon_k,j}-b_j|\}
 \longrightarrow0.
\]
Since
\[
 \sup_{E\in\mathcal S_0}\dist(E,U_0)
 \le d_{\rm H}(\mathcal S_0,\mathcal S_{\varepsilon_k})
       +d_{\rm H}(U_{\varepsilon_k},U_0)\longrightarrow0,
\]
we have $\sup_{E\in\mathcal S_0}\dist(E,U_0)=0$.
Since $U_0$ is closed, $\mathcal S_0\subset U_0$, and
\[
 \sum_{j=1}^M(b_j-a_j)
 =\lim_{k\to\infty}\sum_{j=1}^M(b_{\varepsilon_k,j}-a_{\varepsilon_k,j})
 \le\frac{Cs}{q_n^2}.
\]
which proves the proposition for $d=0$.
\end{proof}

\subsection{Spectral covers on ${\mathcal I}_{\rm iso}$.}
\label{sec:isotropic-cover}
We now take $\lambda\in{\mathcal I}_{\rm iso}$, so $v(x)$ and $b(x)$ are given 
by \eqref{vb-iso}.
In this case we have the following commutator matrix $K_N$ 
{\it independent} of phase. 
Define
\be\label{Kiso}
 \begin{split}
 \varphi&=\arccos\frac1{2\lb_1},\qquad
 \kappa_r=\frac{\pi(-1)^{r-1}\cos(r\varphi)}{\sin(\pi\al r)},\qquad r\ge1,\\
 (K_N)_{j,j+r}&=\kappa_r,\qquad
 (K_N)_{j+r,j}=-\kappa_r,\qquad (K_N)_{jj}=0,
 \quad 0\le j<j+r\le N-1.
 \end{split}
\ee

\begin{lemma}\label{lemmaEKiso}
Let $N>2$, $B_N(x)=\wh B_N(x)$ be given by \eqref{BhN} with $v(x)$, $b(x)$ as defined in \eqref{vb-iso}, and 
let $K_N$ be as defined in \eqref{Kiso}, $t_j$ in \eqref{tj}.
Then 
\[
 \mathcal E_N(x):=\frac{d}{dx}B_N(x)-[K_N,B_N(x)]
 \]
has the following edge matrix elements for $1\le j\le N-2$,
\be\label{Eedgeiso}\begin{aligned}
(\mathcal E_N)_{0,j}&=(\mathcal E_N)_{j,0}=\frac{(-1)^{j+1}\pi}{\sin(\pi(j+1)\al)}\cos((j+1)\varphi)b(t_{-1}),\\
(\mathcal E_N)_{j,N-1}&=(\mathcal E_N)_{N-1,j}=\frac{(-1)^{N-1-j}\pi}{\sin(\pi(N-j)\al)}\cos((N-j)\varphi)b(t_{N-1}),\\
\end{aligned}
\ee
and the corner elements
\be\label{Ecorneriso}\begin{aligned}
(\mathcal E_N)_{0,0} &=-\frac{2\pi}{\sin(\pi\al)}\cos( \varphi)b(t_{-1}),\\
(\mathcal E_N)_{N-1,N-1}&=\frac{2\pi}{\sin(\pi\al)}\cos(\varphi)b(t_{N-1}),\\
(\mathcal E_N)_{0,N-1}&= -\kappa_{N-1}(v(t_{N-1})-v(t_0))-
 \kappa_{N-2}(b(t_{N-2})-b(t_0)).
\end{aligned}
\ee
and the rest of the matrix elements of $\mathcal E_N(x)$ are zero.
\end{lemma}

\begin{proof}
Similarly to Lemma \ref{lemmaEK}, we see that ${\mathcal E}_N$ is symmetric,
so it is sufficient to verify its upper-right corner.
We proceed in the same way as in the proof of Lemma \ref{lemmaEK}
and use the relations \eqref{commut}. The argument is simpler in this case. In particular, for $1\le j\le N-2$, $\ell\ge 2$, \eqref{commut} with the present $B$ and $K$ gives
\be\begin{aligned}
[K,B]_{j,j+\ell}
&=\frac{(-1)^{\ell-1}\pi}{\sin\pi\ell\al}\cos(\ell\varphi)(v(t_{j+\ell})-v(t_j))+
\frac{(-1)^{\ell}\pi}{\sin\pi(\ell-1)\al}\cos((\ell-1)\varphi)(b(t_{j+\ell-1})-b(t_j))\\
&+
\frac{(-1)^{\ell}\pi}{\sin\pi(\ell+1)\al}\cos((\ell+1)\varphi)(b(t_{j+\ell})-b(t_{j-1})).
\end{aligned}
\ee
Substituting here the explicit formulae for $v(t)$ and $b(t)$ and using the fact that $\cos\varphi=1/(2\lb_1)$, we obtain that $[K,B]_{j,j+\ell}=0$.
The other matrix elements are verified similarly.
\end{proof}

\begin{proposition}\label{prop:isotropic-cover}
Let $\lambda\in\mathcal I_{\rm iso}$, $\alpha$ be irrational, and $n\ge4$. There are closed intervals $I_{n,1},\ldots,I_{n,M_n}$ such that
\be
 \Sigma_{\alpha,\lambda}\subset\bigcup_{j=1}^{M_n}I_{n,j},\qquad
 M_n\le q_n+q_{n-1},\qquad
 \sum_{j=1}^{M_n}|I_{n,j}|\le\frac{C\lb_1}{q_n},
\ee
where $C>0$ is absolute.
\end{proposition}
\begin{proof}
Applying similar arguments as in the proof of Lemma \ref{lemmaEKnorms}
to $\mathcal E_N$ of Lemma \ref{lemmaEKiso}, we easily see that
$\|\mathcal E_N\|_{S_2}\le C\lb_1$,
$\|[K_N,P_N]\|_{S_2}\le Cq_n$,
and therefore Propositions  \ref{prop:blocks}, \ref{prop:block-cover}
with $J=J_n$, $K=K_N$, $\mathcal E=\mathcal E_N$, 
$\tau=\bt=\sup_k|b(x_k)|\le C\lb_1/q_n$, give the desired results.
\end{proof}

\begin{proof}[Proof of Theorem~\ref{thm:all-scales}]
This is the union of Propositions \ref{prop:L123-cover}, \ref{prop:isotropic-cover}.
\end{proof}

\section{Dimension bounds}\label{sec:dimensions}

The covers in Section~\ref{sec:covers} give upper bounds for fractal dimensions
of the spectrum. For $\beta(\alpha)>0$ (see \eqref{beta}), we combine them with lower bounds
on packing dimension from \cite{HYZ}.

\subsection{Dimension bounds from spectral covers: proof of
Theorems \ref{thm:HD} and \ref{thm:upper-box}}\label{sec:dimension-upper}

For a nonempty compact set $\Omega\subset\R$, we write $N_\Omega(\varepsilon)$ for the least number of intervals of length at most $\varepsilon$ covering $\Om$. 
Then the lower and upper box dimensions of $\Om$ are given, respectively, by
(see, e.g., \cite{falc})
\[
 \lbox \Om=\liminf_{\varepsilon\downarrow0}\frac{\ln N_\Om(\varepsilon)}{\ln(1/\varepsilon)},
 \qquad \ubox \Om=\limsup_{\varepsilon\downarrow0}\frac{\ln N_\Om(\varepsilon)}{\ln(1/\varepsilon)}.
\]
When these coincide, their common value is the box-counting dimension $\dim_{\mathrm B}\Om$. 
The $t$-dimensional Hausdorff measure of $\Om$ is
\[
 \meas_t(\Om)=\lim_{\delta\downarrow0}\inf\left\{\sum_j |I_j|^t:\Om\subset\bigcup_jI_j,\ |I_j|\le\delta\right\},
\]
where infimum is taken over all finite or countable covers of $\Om$ by intervals $I_j$
of length no larger than $\de$.
The Hausdorff dimension of $\Om$ is then $\dimH(\Om)=\inf\{t>0:\meas_t(\Om)<\infty\}$. The packing dimension is
\[
 \pdim \Om=\inf_{\Om\subset\bigcup_{j\ge1}\Om_j}\sup_j\ubox \Om_j,
\]
where the covering sets $\Om_j$ are bounded. 
Recall that in contrast with box dimensions, the packing dimension is
countably stable: $\pdim\cup_{j=1}^\infty\Om_j=\sup_j\pdim\Om_j$.
Furthermore, recall that the dimensions satisfy the following inequalities:
\be\label{dimcompare}
\dimH\Om\le\lbox\Om\le\ubox\Om,\qquad \dimH\Om\le\pdim\Om\le\ubox\Om.
\ee

\begin{proposition}\label{prop:dimension}
Let $\Om\subset\R$ be nonempty and compact, $Q_n\uparrow\infty$ 
monotonically
as $n\to\infty$, and 
$\varkappa>0$. Suppose that, for $n\ge n_0$ with some $n_0>0$, 
there are covers of $\Om$ by intervals $I_{n,j}$ satisfying
\[
\Om\subset\bigcup_{j=1}^{M_n}I_{n,j},\qquad M_n\le AQ_n,\qquad \sum_j|I_{n,j}|\le BQ_n^{-\varkappa},
\]
for some $A,B>0$ independent of $n$.
Let $d_*=1/(1+\varkappa)$. Then
\begin{equation}\label{eq:H-measure-lbox-ineq}
 \meas_{d_*}(\Om)\le A^{1-d_*}B^{d_*},\qquad  \lbox \Om\le d_*
\end{equation}
Moreover,
\be\label{uboxest}
\ubox \Om\le\limsup_n\frac{\ln Q_{n+1}}{\varkappa\ln Q_n+\ln Q_{n+1}}.
\ee
\end{proposition}
\begin{proof}
For $0<t\le1$, H\"older inequality gives
\[
 \sum_{j=1}^{M_n}|I_{n,j}|^t\le M_n^{1-t}\left(\sum_{j=1}^{M_n}|I_{n,j}|\right)^t
 \le A^{1-t}B^t Q_n^{1-(1+\varkappa)t}.
\]
Since $\max_j|I_{n,j}|\le BQ_n^{-\varkappa}\to0$, taking $t=d_*$ 
shows that $\meas_{d_*}(\Om)\le A^{1-d_*}B^{d_*}$. This immediately implies that 
$\dimH\Om\le d_*$. However, we will now obtain a stronger (in view of \eqref{dimcompare})
result on $\lbox\Om$.

Dividing the intervals $I_{n,j}$ into subintervals of length $\varepsilon$, we obtain
\be\label{Nep}
 N_\Om(\varepsilon)\le \sum_{j=1}^{M_n}\left(\frac{1}{\ep}|I_{n,j}|+1\right)=
 \frac{1}{\varepsilon}\sum_{j=1}^{M_n}|I_{n,j}|+M_n
 \le\frac{B}{\ep Q_n^{\varkappa}}+AQ_n,\qquad n\ge n_0. 
\ee
Let
\be\label{epj}
\ep_j=\frac{1}{Q_j^{1+\varkappa}},\qquad j\ge n_0.
\ee
We take $n_0$ sufficiently large so that $\ep_{n_0}<1$.
Setting $\ep=\ep_n$ in \eqref{Nep} we obtain that 
$
N_\Om(\ep_n)\le(A+B)Q_n=(A+B)\ep_n^{-d_*},
$
which immediately implies that $\lbox\Om\le d_*$.

It remains to estimate $\ubox\Om$. Let $\ep_{n+1}\le\ep\le\ep_n$. First, we note that by \eqref{Nep},
\be\label{estN1}
N_\Om(\ep)\le \frac{B}{\ep Q_{n+1}^{\varkappa}}+AQ_{n+1}
\le \frac{B}{\ep_{n+1} Q_{n+1}^{\varkappa}}+AQ_{n+1}=(A+B)Q_{n+1}.
\ee
On the other hand, since
\[
Q_n=\frac{1}{\ep_n^{1/(1+\varkappa)}}\le\frac{1}{\ep^{1/(1+\varkappa)}}=
\frac{\ep^{\varkappa/(1+\varkappa)}}{\ep}\le 
\frac{\ep_n^{\varkappa/(1+\varkappa)}}{\ep}=\frac{1}{\ep Q_n^\varkappa},
\]
we also have
\be\label{estN2}
N_\Om(\ep)\le \frac{B}{\ep Q_n^{\varkappa}}+AQ_n\le 
\frac{A+B}{\ep Q_n^\varkappa}.
\ee
Together, \eqref{estN1} and \eqref{estN2} give
\be\label{estN3}
N_\Om(\ep)\le (A+B)\min\left( \frac{1}{\ep Q_n^\varkappa},Q_{n+1}\right),\qquad
\ep_{n+1}\le\ep\le\ep_n.
\ee
Denote $x=\ln(1/\ep)$, $x_*=\varkappa\ln Q_n+ \ln Q_{n+1}$. Note that
\[
\ln\frac{1}{\ep_n}=(1+\varkappa)\ln Q_n\le x_*
\le (1+\varkappa)\ln Q_{n+1}=\ln\frac{1}{\ep_{n+1}},
\]
and
\[
f(x):=\frac{1}{x}\min(x-\varkappa\ln Q_n,\ln Q_{n+1})= 
 \begin{cases}
 1-\varkappa \ln Q_n/x,&x\le x_*,\\ \ln Q_{n+1}/x,&x\ge x_*.
 \end{cases}
\]
Clearly, this function reaches its maximum on $(0,\infty)$ at $x=x_*$, and we have
\[
f(x)\le\frac{\ln Q_{n+1}}{\varkappa\ln Q_n+ \ln Q_{n+1}},\qquad
 \ep_{n+1}\le\ep\le\ep_n.
\]
Therefore \eqref{estN3} implies
\be
\frac{\ln N_\Om(\ep)}{\ln 1/\ep}\le
\frac{\ln (A+B)}{\ln 1/\ep}+f(x)\le \frac{\ln (A+B)}{\ln 1/\ep}+
\frac{\ln Q_{n+1}}{\varkappa\ln Q_n+ \ln Q_{n+1}},
\ee
and taking $\limsup$ gives \eqref{uboxest}.
\end{proof}

\begin{proof}[Proof of Theorems \ref{thm:HD},\ref{thm:upper-box}]
We apply Theorem \ref{thm:all-scales} together with Proposition~\ref{prop:dimension} where $\Om=\Sigma_{\alpha,\lambda}$ with 
$\lambda\in\mathcal C_0$, and $Q_n=q_n$.
By Theorem \ref{thm:all-scales} we can take
$\varkappa=1$, improving to $\varkappa=2$ on $\mathcal B_{\rm iso}$. 
This with \eqref{dimcompare} proves Theorem~\ref{thm:HD}. 
Theorem~\ref{thm:upper-box} follows from \eqref{uboxest} since
for positive $C$, 
the function
$t/(C+t)$, $t>0$, is continuous and strictly increasing.
\end{proof}

\subsection{Packing lower bounds and absence of box-counting dimension: Proof of Corollary \ref{cor:box-nonexistence}.}\label{app:packing}
For a finite nonzero measure $\mu$ on $\R$, its packing dimension is
\[
 \dim_{\rm P}\mu=\inf\{\pdim F:\mu(\R\setminus F)=0\}.
\]

For a bounded self-adjoint operator $H$ on a Hilbert space, the spectral measure associated with $f\ne0$, is defined in terms of the inner product as 
$\mu_f(B)=(f,\mathbf1_B(H)f)$, where $B$ is a Borel set in $\R$ and 
$\mathbf1_B$ its indicator function. The spectral measure satisfies
\[
\supp\mu_f\subset\spec(H),\qquad \dim_{\rm P}\mu_f\le\pdim\spec(H)\le 1.
\]

We assume $\beta(\alpha)>0$ and take the parameters in Corollary~\ref{cor:box-nonexistence}:
\[
 \lambda\in\{\lambda_{\mathrm A},\lambda_{\mathrm T}\}
 \cup\bigl((\mathcal I_{\mathrm{iso}} { \cup \mathcal B_{\mathrm{iso}}})\cap\{\lambda_2>0\}\bigr).
\]
By \cite[Corollary~2.6(3) and (2.18)]{HYZ}, 
in this case there is a spectral measure $\mu_f$ of $H_{\al,\th,\lb}$ 
such that $\dim_{\rm P}\mu_f=1$. Therefore, $\pdim \Sigma_{\alpha,\lambda}=1$,
and hence (by \eqref{dimcompare}) also  $\ubox\Sigma_{\alpha,\lambda}=1$.
For the almost Mathieu operator, the equality
$\pdim \Sigma_{\alpha,\lambda}=1$
also follows from \cite[Corollary~4]{JZ}.
By 
Theorem~\ref{thm:HD} we then obtain \eqref{noboxdim}.
In particular, we see that the lower and upper box dimensions are different so that there is no box-counting dimension.
This proves Corollary \ref{cor:box-nonexistence}.

\end{document}